\documentclass[conference]{IEEEtran}
\makeatletter
\def\ps@headings{%
\def\@oddhead{\mbox{}\scriptsize\rightmark \hfil \thepage}%
\def\@evenhead{\scriptsize\thepage \hfil \leftmark\mbox{}}%
\def\@oddfoot{}%
\def\@evenfoot{}}
\makeatother
\ifCLASSINFOpdf
\else
\fi
\usepackage[cmex10]{amsmath}
\usepackage{amsfonts}
\usepackage{amssymb}
\usepackage{bm}
\newtheorem{lemma}{Lemma}
\newtheorem{cor}{Corollary}
\newtheorem{thm}{Theorem}
\allowdisplaybreaks[1]

\begin{document}
\title{Rateless Resilient Network Coding Against Byzantine Adversaries}

%\author{\IEEEauthorblockN{Michael Shell}
%\IEEEauthorblockA{School of Electrical and\\Computer Engineering\\
%Georgia Institute of Technology\\
%Atlanta, Georgia 30332--0250\\
%Email: http://www.michaelshell.org/contact.html}
%\and
%\IEEEauthorblockN{Homer Simpson}
%\IEEEauthorblockA{Twentieth Century Fox\\
%\and
%\IEEEauthorblockN{James Kirk\\ and Montgomery Scott}
%\IEEEauthorblockA{Starfleet Academy\\
%San Francisco, California 96678-2391\\

\author{\IEEEauthorblockN{Wentao Huang, Tracey Ho, Hongyi Yao}
\IEEEauthorblockA{California Institute of Technology, USA\\\{whuang, tho, hyao\}@caltech.edu}
\and
\IEEEauthorblockN{Sidharth Jaggi}
\IEEEauthorblockA{Chinese University of Hong Kong, Hong Kong\\jaggi@ie.cuhk.edu.hk  }
}


%Telephone: (800) 555--1212, Fax: (888) 555--1212}


\maketitle

\begin{abstract}
This paper considers rateless network error correction
codes for reliable multicast in the presence of adversarial
errors. Most existing network error correction codes are designed
for a given network capacity and maximum number of errors
known a priori to the encoder and decoder. However, in certain
practical settings it may be necessary to operate without such
a priori knowledge. We present rateless coding schemes for two
adversarial models, where the source sends more redundancy
over time, until decoding succeeds. The first model assumes
there is a secret channel between the source and the destination
that the adversaries cannot overhear. The rate of the channel
is negligible compared to the main network. In the second
model, instead of a secret channel, the source and destination
share random secrets independent of the input information. The
amount of secret information required is negligible compared to
the amount of information sent. Both schemes are optimal in that
decoding succeeds with high probability when the total amount of
information received by the sink satisfies the cut set bound with
respect to the amount of message and error information. The
schemes are distributed, polynomial-time and end-to-end in that
other than the source and destination nodes, other intermediate
nodes carry out classical random linear network coding.
\end{abstract}


% \fi
\IEEEpeerreviewmaketitle

\section{Introduction}
Network coding is a technique that allows the mixing of data at intermediate network nodes instead of simply relaying them. It has been shown theoretically that such codes are capable of achieving multicast network capacity and can be implemented in a distributed manner, as well as improving robustness against packet losses and link failures \cite{Ahlswede2000,Jaggi2005,Koetter2003,Ho2003}.
However, comparing with pure forwarding of packets, network coding is more vulnerable to attack by malicious adversaries that inject corrupted packets, since a single corrupted packet is mixed with other packets in the network.
%The impact of Byzantine adversaries on network coding has been studied extensively in the literature. Ho et al. \cite{Ho2004} discuss the detection of such adversaries. Other works study practical solution to correct the errors based on homomorphic hashing \cite{Gkantsidis2006,Krohn2004}, or signatures \cite{Charles2006} for computationally security under various threat models. Capacity bounds in such adverse environments are investigated in
The use of coding to correct such errors information theoretically was introduced by~\cite{Yeung2006,Cai2006}, and capacity-achieving network error correction codes have been proposed for various adversary and network models, e.g.~\cite{koetter08coding,Jaggi2008,Nutman2008}.  However, most existing schemes  assume a given min cut (capacity) of the network and maximum number of adversarial errors for the purposes of code design and encoding. But such an assumption may be overly restrictive in many practical settings. For example, in large peer-to-peer content distribution networks, estimating  network capacity is not easy, and the capacity is likely to change over time as users join and leave the network. Furthermore, it would be even more difficult to decide the number of malicious nodes and their strength. This issue becomes more serious if the source is multicasting to many destinations, where different destinations may require different code constructions to suit their own parameters.%, but such multiple unicasts would be a severe waste of bandwidth.

This paper proposes rateless network error correction codes that do not require an a priori estimates of the network capacity and number of errors.
%As in fountain codes proposed for the erasure decoding, The encoder of a fountain code \cite{MacKay2005} behaves like a fountain with an endless supply of water drops, and anyone who wants to receive the data holds a bucket to collect them. One can decode successfully once the amount of drops he collects is larger than a certain threshold. This class of codes is rateless in the sense that
The source transmits redundancy incrementally until decoding succeeds. The supply of encoded packets is potentially limitless and the number of encoded packets actually transmitted is determined by the number of errors that occur. A number of related works e.g.~\cite{Krohn2004,Gkantsidis2006,Charles2006} propose cryptographic schemes that can be used to detect and remove errors in rateless network codes, while~\cite{Vyetrenko2009} proposes a rateless network error correction scheme that requires cryptographic means of verifying successful decoding.  In contrast, our work presents the first completely information-theoretic rateless network error correction codes.

We design two algorithms targeting different network models. In the first model, also studied in \cite{Jaggi2008}, there is a secret channel between the source and the destination that is hidden from the adversary (who is omniscient except for the secret), and the rate of the channel is negligible compared to the network. In this case over time we incrementally send more linearly dependent redundancy of the source message through the network to combat erasures, and incrementally send more (linearly independent) short hashes of the message on the secret channel to eliminate fake information. The destination amasses both kinds of redundancy until he decodes successfully. The code will adapt to the actual min cut of the network as well as the number of errors. %, and achieve an average rate of $M-z$, where $M$ is the expected capacity of the network and $z$ is the expected number of errors injected. This rate is optimal given the strength of the adversaries.

The second scenario is the random secret model \cite{Nutman2008}, where  instead of a secret channel, the source and destination share a ``small'' fixed random secret that is independent of the input message. The amount of secrets required is again negligible compared to the amount of information sent. The random secret model may be more realistic than the secret channel model because it allows the source and destination to share their secrets in advance and use them for later communication over time. It is also possible for source and destination to share only a secret seed and generate pseudo random sequences with the seed \cite{Blum1982}. Compared to the secret channel model, the challenge is that both linearly dependent and independent redundancy must be sent over the public and unreliable network. Again, our code will adapt to the network and adversary parameters. %and achieves the optimal rate $M-z$ under this model.

Both schemes  are distributed with polynomial-time complexity of design and implementation. They assume no knowledge of the topology and work in both wired and wireless networks. Moreover, implementation involves only slightly modifying the source encoder and destination decoder, while internal nodes use standard random linear network coding.

\section{Network Models}\label{secmodel}
\subsection{Adversary Model}
The source Alice wishes to communicate reliably with the destination Bob over a general network, where there is a hidden attacker Calvin who wants to disrupt the communication. Calvin is assumed to be able to observe all the transmissions over the network, and know the encoding and decoding schemes at Alice, Bob, as well as all other intermediate nodes. He is also aware of the network topology. Calvin can corrupt transmitted packets or inject erroneous packets. Finally, we assume Calvin to be computationally unbounded, so information-theoretic security is required in this case.

However, we assume that Calvin's knowledge is limited in some aspects. Two limitation models are discussed in this paper. For the first model, in addition to the given network, we assume there is a secret channel between Alice and Bob, \emph{i.e.}, the information transmitted on this channel will not be observed or modified by Calvin \cite{Jaggi2008}. However, the rate of the channel is negligible compared to the network. In the second model, we assume the source and destination share a small amount of random secret information that is independent with the input information \cite{Nutman2008}. Again, the amount of secret information required is negligible compared to the amount of information sent. As we will show later, the differences between the two models and the respective code constructions are substantial.

\subsection{Network Model}
We model the network in the general case as a hypergraph where nodes are vertices and hyperedges
are directed from the transmitting nodes to the set of the receiving nodes \cite{Lun2005}. Let $\mathcal{E}$ be the set of hyperedges and $\mathcal{T}$ be the set of nodes. Alice and Bob are not assumed to have knowledge of the topology of the hypergraph. They also may not know the capacity of the network as well as the number of errors that the adversary can inject.

Source Alice encodes her information bits into a batch of $b$ packets by the encoding schemes described in subsequent sections. Each packet contains a sequence of $n+b$ symbols from the finite field $\mathbb{F}_q$. Let matrix $X_0=\mathbb{F}_q^{b \times (n+b)}$ be one batch of packets from Alice that is desired to be communicated to Bob. We call the successful communication of one batch of information bits $X_0$ to the destination a session. For clarity, we focus on one such session. In the rateless setting, because Alice does not know the network capacity and error patterns, a session may require multiple network transmissions until Bob receives enough redundancy to cope with errors and decode correctly. Assume in general that a session involves $N$ stages, \emph{i.e.}, $N$ uses of the network, where $N$ is a variable. During the $i$-th stage, $1 \le i \le N$, denote the capacity (min cut from Alice to Bob) of the network as $M_i$, and the number of errors (min cut from Calvin to Bob) that the adversary injects as $z_i$. We assume $z_i < M_i$, otherwise the network capacity is completely filled with errors and it is not possible to transmit anything. For any realistic network, $M_i$ is always bounded. For example, let $c_i$ be the number of transmission opportunities that occur to the source during the $i$-th stage, then $M_i \le c_i$. For convenience we further assume $c_i \le \bar{c},\ \forall i$.

%and we assume $M_i \le \bar{c}$, $z_i \le \bar{z}$, and $\bar{z}<\bar{c}$, $\forall i$ with high probability.

\section{Code Construction for Secret Channel Model }
%For a linear code (which suffice for multicase [\ref{sadf}]), we have $Y_i = T_i X_i$, where $T_{i} \in \mathbb{F}_q^{_i \times b}$ is the transition matrix of the network.
\subsection{Encoder}\label{secenc}
Alice's encoder has a structure similar to the secret channel model in \cite{Jaggi2008}, but operates in a rateless manner. In each session Alice transmits $nb$ incompressible information symbols from $\mathbb{F}_q$ to Bob. Alice arranges them into a matrix $W \in \mathbb{F}_q^{b \times n}$, and encodes $X_0 = (W\ I_b)$, where $I_b$ is the identity matrix of dimension $b$. Then, as in \cite{Ho2006,Ho2003}, Alice performs random linear combinations to the rows of $X_0$ to generate her transmitted packets. Specifically, Alice draws a random matrix $K_1 \in \mathbb{F}_q^{c_1 \times b}$ and encodes $X_1 = K_1 X_0$. Note that the redundant identity matrix receives the same linear transform so we can recover $X_0$. $X_1$ is then send over a network  where intermediate nodes implement random linear coding. In addition, Alice will hash the message and send it through the secret channel. She sets $\alpha_1 = b c_1 $, and draws random symbols $r_1, ..., r_{\alpha_1+1}$ independently and uniformly from $\mathbb{F}_q$. Note that the $\{r_j\}$ are drawn secretly so that Calvin cannot observe them. Let $D_1=[d_{kj}] \in \mathbb{F}_q^{(n+b) \times (\alpha_1+1)}$, where $d_{kj} = (r_j)^k$, and then the hash is computed as $H_1 = X_0 D_1$. Finally Alice sends $r_1, ..., r_{\alpha_1+1}$ and $H_1$ to Bob through the secret channel. The size of the secret is $(\alpha_1+1)(b+1)$, which is asymptotically negligible in $n$.

Alice then keeps sending more redundant information to Bob as follows.
For the $i$-th stage, $i \ge 2$, Alice draws a random matrix $K_i \in \mathbb{F}^{c_i \times b}_q$, encodes $X_i = K_i X_0$, and sends $X_i$ over the network. In addition, Alice again draws $r_1, ..., r_{\alpha_i}$ randomly from $\mathbb{F}_q$ secretly, where $\alpha_i=b c_i$. She then constructs $D_i = [d_{kj}] \in \mathbb{F}^{(n + b) \times \alpha_i}_q$, $d_{kj}=(r_j)^k$, and computes $H_i = X_0 D_i$. Alice eventually sends $r_1, ..., r_{\alpha_i}$ and $H_i$ to Bob through the secret channel. The size of the secret is $\alpha_i(b+1)$, again asymptotically negligible in $n$. Alice repeats this procedure until Bob indicates decoding success. If a success is indicated, Alice ends the current session and moves onto the next session.

\subsection{Decoder}
The network performs a classical distributed network code (which is shown to suffice to achieve capacity for multicast \cite{Ho2003}). Specifically, each packet transmitted by an intermediate node is a random linear combination of its incoming packets. For the $i$-th stage, we can describe this linear relation as
\begin{equation*}
  Y_i = [T_i \  Q_i] \left[ \begin{array}{c}
    X_i \\
    Z_i
  \end{array} \right]
,\end{equation*}
where $Y_i \in \mathbb{F}_q^{M_i \times (n+b)}$ is Bob's received observation, $Z_i \in \mathbb{F}_q^{z_i \times (n+b)}$ is the errors injected by Calvin, and $T_i$ and $Q_i$ are defined to be the transfer matrix from Alice to Bob and from Calvin to Bob, respectively.
By stacking all the batches of observations received by the $i$-th stage, let
\begin{equation*}
  Y^{(i)} = \left[ \begin{array}{c}
    Y_1\\
    \vdots\\
    Y_i
  \end{array} \right], \ \ \
  Z^{(i)} = \left[ \begin{array}{c}
    Z_1\\
    \vdots\\
    Z_i
  \end{array} \right],
\end{equation*}
\begin{equation*}
      H^{(i)} = \left[ H_1\  ...\  H_i \right],\ \ \  D^{(i)} = \left[D_1 ... D_i\right],
\end{equation*}
\begin{equation*}
    \hat{T}^{(i)} = \left[ \begin{array}{c|cccc}
    T_1 K_1& Q_1 & 0 & ... & 0\\
    T_2 K_2 & 0 & Q_2 & ... & 0\\
    \vdots & \vdots\\
    T_i K_i & 0 & 0 & ... & Q_i
  \end{array} \right] =
  \left[ T^{(i)}\ | \ Q^{(i)} \right].
\end{equation*}
Then we have
\begin{align}\label{first}
Y^{(i)} &= \left[ T^{(i)}\ \ Q^{(i)} \right] \left[\begin{array}{c}
  X_0\\
  Z^{(i)}
\end{array}\right],\\\label{sec}
  X_0 D^{(i)} &= H^{(i)},
\end{align}
where (\ref{first}) follows from the network transform, and (\ref{sec}) follows from the code construction. Notice that only $Y^{(i)}$, $D^{(i)}$ and $H^{(i)}$ are available to Bob, and he needs to recover $X_0$ from equations (\ref{first}), (\ref{sec}). Bob can accomplish this only if $X_0$ is in the row space of $Y^{(i)}$.% But in order to avoid being deceived by the corrupted information (discussed later), Alice need to make sure that the number of rows of $Y^{(i)}$ follows $\sum_{j=1}^i M_j \le i \bar{c}$.  By the definition of $\bar{c}$ this should always holds. However, in case that we do not have an accurate estimate of $\bar{c}$, \emph{i.e.}, if $\sum_{j=1}^i M_j > i \bar{c}$ for some $i$, then Bob can extract any arbitrary $u_i=i \bar{c}$ rows of $Y^{(i)}$ and arrange them into a submatrix $\tilde{Y}^{(i)} \in \mathbb{F}_q^{u_i \times (n+b)}$. Then he use $\tilde{Y}^{(i)}$ instead of $Y^{(i)}$ to decode thereafter, following the same procedure described below.

Suppose $X_0$ indeed lies in the row space of $Y^{(i)}$ (it happens with high probability for some $i$ as shown later), then there exists $X^s$ such that
\begin{equation}\label{x0}
  X_0= X^s Y^{(i)}.
\end{equation}
Therefore Bob only needs to find $X^s$, which may be achieved by solving
\begin{equation}\label{xs}
  X^s {Y}^{(i)} D^{(i)} = H^{(i)}.
\end{equation}
If (\ref{xs}) has a unique solution for $X^s$, Bob reconstructs $X_0$ according to (\ref{x0}) and feedbacks an acknowledgement of decoding success to Alice. If there exists no solution for (\ref{xs}), Bob waits for more redundancy to come. Otherwise, if there are multiple solutions for (\ref{xs}), Bob declares a decoding failure.
\subsection{Performance}
In the following we will show that the probability of error, including the events that Bob declares an error, or the events that $X_0$ is not in the row space of ${Y}^{(i)}$, or there exists some other $X' \ne X_0$ that satisfies (\ref{xs}), is vanishing as $q \to \infty$.

We first show that under proper conditions, $X_0$ is in the row span of $Y^{(i)}$. The following lemma is well-known \cite{Naylor94}:
\begin{lemma}\label{firstlem}
  If the linear transform $\hat{T}^{(i)}$ has full column rank, \emph{i.e.}, Rank($\hat{T}^{(i)}$)$=b+\sum_{j=1}^i z_j$, then it is left-invertible, and there exists $X^s$ such that $X^s Y^{(i)} = X_0$.
\end{lemma}

Now we show that $\hat{T}^{(i)}$ almost always has either full column rank or full row rank as $q \to \infty$.
\begin{lemma}\label{disjoint}
  If $b+\sum_{j=1}^i z_j \le \sum_{j=1}^i M_j$, then $\hat{T}^{(i)}$ has full column rank with high probability\footnote{Event $E$ happens with high probability ($w.h.p.$) if $\lim_{q \to \infty}\Pr \{E\}=1$.}.
\end{lemma}
\begin{proof}
The proof follows an idea similar to \cite{Jaggi2008} with the difference that we consider communications over multiple stages.
  Notice that $\hat{T}^{(i)} =\left[ T^{(i)}\ \ Q^{(i)} \right]$. Since $b+\sum_{j=1}^i z_j \le \sum_{j=1}^i M_j$,  it follows $b \le \sum_{j=1}^i M_j$. Because $T_j$ has full rank and $K_j$ are random matrices, if $b \le \sum_{j=1}^i M_j$, then the probability that the columns of $T^{(i)}$ are linearly dependent is upper bounded by $b/q \to 0$ by the Schwartz-Zippel lemma. So $T^{(i)}$ has full column rank. Without loss of generality we assume $Q^{(i)}$ also has full column rank, otherwise we can select a basis of the column space of $Q^{(i)}$ and reformulate the problem with a reduced $z_i$. Furthermore, by \cite{Ho2003al}, if $b+\sum_{j=1}^i z_j \le \sum_{j=1}^i M_j$, the probability that the column spans of $T^{(i)}$ and $Q^{(i)}$ intersects anywhere other than in the zero vector is upper bounded by $i^2|\mathcal{T}||\mathcal{E}|q^{-1}$ for a fix adversary pattern. Since Calvin can choose his locations in at most $i|\mathcal{E}| \choose \sum_j z_j$ ways, by the union bound, the probability that $T^{(i)}$ and $Q^{(i)}$ intersects is bounded by ${i|\mathcal{E}| \choose \sum_j z_j} i^2|\mathcal{T}||\mathcal{E}|q^{-1} \to 0$. Hence $\hat{T}^{(i)}$ has full column rank with high probability.
\end{proof}

%Corollary \ref{xse} below follows from Lemma \ref{firstlem} and \ref{disjoint}.
\begin{cor}\label{xse}
  If $b+\sum_{j=1}^{i} z_j \le \sum_{j=1}^{i} M_j$, then with high probability there exists a $X^s$ such that $X^s {Y}^{(i)} = X_0$.
\end{cor}

%Hence we can assert that (\ref{x0}) holds for sufficiently large $i$.
Next we need to show the solution is unique, \emph{i.e.}, the hash is strong enough so that Bob can distill the injected error.
\begin{lemma}\label{hash}
  For any $X' \ne X_0$, the probability that $X' D^{(i)} = H^{(i)}$ is bounded from above by $((n+b)/q)^{\sum_{k=1}^i \alpha_k+1}$.
\end{lemma}
\begin{proof}
     It is equivalent to consider the probability that $(X'-X_0) D^{(i)} = 0$. Since $X' - X_0 \ne 0$, there is at least one row in which $X'$ differs from $X_0$. Denote this row of $X'-X_0$ as $(x_1, ..., x_{n+b})$, then the $j$-th entries of the corresponding row of $(X'-X_0)D^{(i)}$ is $F(r_j)=\sum_{k=1}^{n+b}x_k r_j^k$. Because $F(r_j)$ is not the zero polynomial, the probability (over $r_j$) that $F(r_j)=0$ is at most $(n+b)/q$.  Because $D^{(i)}$ has $\sum_{k=1}^i \alpha_k+1$ columns, and all $r_j$, $1 \le j \le \sum_{k=1}^i \alpha_k+1$, are independently chosen, the probability that the entire row is a zero vector is at most $((n+b)/q)^{\sum_{k=1}^i \alpha_k+1}$. This is an upper bound of the probability that the entire matrix of $(X'-X_0) D^{(i)}$ is zero.
\end{proof}
\begin{lemma}\label{vsne}
  The probability that there exists ${V^s} \ne X^s$ such that $V^s {Y}^{(i)} \ne X_0$ but $V^s {Y}^{(i)} D^{(i)} = H^{(i)}$ is upper bounded by $(n+b)^{\sum_{k=1}^i \alpha_k+1}/q \to 0$.
\end{lemma}
\begin{proof}
  Note that the dimension of $V^s$ is $$b  \sum_{j=1}^{i} M_j \le b  \sum_{j=1}^i c_j$$ over $\mathbb{F}_q$. So by invoking Lemma \ref{hash} and then take the union bound over all possible choices of $V^s$, the claim follows.
\end{proof}

Now we are ready to present the main result for the shared secret model.
\begin{thm}\label{thm1}
  $\forall i$ such that $b+\sum_{j=1}^{i} z_j \le \sum_{j=1}^{i} M_j$, with the proposed coding scheme, Bob is able to decode $X_0$ correctly with high probability at the $i$-th stage. Otherwise, Bob waits for more redundancy instead of decoding erroneous packets.
\end{thm}
\begin{proof}
  By Corollary \ref{xse}, we can solve $X_0$ from (\ref{x0}) and (\ref{xs}) if $b+\sum_{j=1}^{i} z_j \le \sum_{j=1}^{i} M_j$. By Lemma \ref{vsne}, if a solution exists, it is correct and unique. Otherwise, there is no solution to (\ref{xs}) and by the algorithm Bob waits for more redundancy.
\end{proof}

Theorem \ref{thm1} shows the code is optimal in that decoding succeeds with high probability whenever the total amount of information received by the sink satisfies the cut set bound, \emph{i.e.}, $b+\sum_{j=1}^{i} z_j \le \sum_{j=1}^{i} M_j.$ If the bound is not satisfied, then  it is not possible for Bob to decode correctly under any coding scheme. %And our scheme decodes successfully with high probability whenever the total amount of
%information received by the sink satisfies the cut set bound with
%respect to the amount of message and error information.
 The following result shows that the code is rate optimal if the network capacity and number of errors are  i.i.d. across stages.\begin{thm}\label{rate} Assume $M_i$, $z_i$, $i=1,2...$ are i.i.d. random variables with mean $\mathbb{E}[M]$ and $\mathbb{E}[z]$, respectively.
 If there exists $\epsilon >0$ such that $\mathbb{E}[M]-\mathbb{E}[z] \ge \epsilon$, then with the proposed coding scheme, Bob is able to decode $X_0$ correctly with high probability in a finite number of stages. Further, on average the code achieves rate $$r \ge \frac{b}{b+\bar{c}-1} \left(\mathbb{E}[M] - \mathbb{E}[z]\right).$$
\end{thm}
\begin{proof}
%We first show the algorithm will end in finite stages.
  Let $L = b/\epsilon$,
  and let random variable $\Delta_j=M_j - z_j$, so $\mathbb{E}[\Delta_j]=\epsilon$ and denote $\text{Var} [\Delta_j]=\sigma^2_\Delta<\infty$. Then for $N \ge L$, $\mathbb{E}[\sum_{j=1}^N \Delta_j] = N\epsilon = b + (N-L)\epsilon$.
  By Chebyshev's inequality,
  \begin{align*}
    \Pr\left\{ \sum_{j=1}^N \Delta_j   \ge b \right\}\\
    &\hspace{-25mm} = \Pr\left\{\sum_{j=1}^N \Delta_j   \ge  \mathbb{E}\left[\sum_{j=1}^N \Delta_j\right] - (N-L)\epsilon\right\} \\
    &\hspace{-25mm}\ge 1 - \Pr\left\{ \left|\sum_{j=1}^N \Delta_j - \mathbb{E}\left[\sum_{j=1}^N \Delta_j\right] \right|
    \ge (N-L)\epsilon  \right\}\\
    &\hspace{-25mm}\ge 1 - \frac{N \sigma^2_\Delta}{(N-L)^2\epsilon^2} \to 1\ \ \text{as }N \to \infty.
  \end{align*}
%  \begin{align*}
%    \sum_{j=1}^{N} M_j - \sum_{j=1}^{N} z_j & \ge N(\mathbb{E}[M]-\mathbb{E}[z]) -\epsilon \\
%    & \ge (L-1) \epsilon \ge b
%  \end{align*}
  So with high probability there exists finite $N \ge L$ such that $b+\sum_{j=1}^{N} z_j \le \sum_{j=1}^{N} M_j$. And by Theorem \ref{thm1}, Bob is able to decode successfully at stage $N$. Next we determine the average rate, let $$N = \min\left\{i : b+\sum_{j=1}^{i} z_j \le \sum_{j=1}^{i} M_j\right\}.$$
  Then the average rate of the code is
  \begin{equation}\label{avr}
    r = \mathbb{E}\left[\frac{b}{N}\right] \ge \frac{b}{\mathbb{E}[N]},
  \end{equation}
  where the last inequality follows from Jensen's inequality. %Now rewrite $N$ equivalently,
  %\begin{align*}
  %  N & = \min \left\{ i: \sum_{j=1}^{i}M_{j} - \sum_{j=1}^{i}z_{j} \ge b \right\}
  %\end{align*}
  Denote $S_\tau = \sum_{j=1}^{\tau}M_{j} - \sum_{j=1}^{\tau} z_{j}$, then $N$ is a stopping time for the random process $S_\tau$. Therefore, by
  the first Wald identity \cite{Grimmett01},
  \begin{align*}
     \mathbb{E}[N] ( \mathbb{E}[M]-\mathbb{E}[z]) & = \mathbb{E}[S_N]\\
     & \le b + \bar{c} -1.
  \end{align*}
  Therefore $$\mathbb{E}[N] \le \frac{b + \bar{c} - 1}{\mathbb{E}[M]-\mathbb{E}[z]}.$$
  Substituting into (\ref{avr}), it follows
  \begin{equation*}
    r \ge \frac{b}{b+\bar{c}- 1} \left(\mathbb{E}[M] - \mathbb{E}[z]\right).
  \end{equation*}
\end{proof}

Notice that as we choose a sufficiently large $b$, then the rate of the code is approaching $\mathbb{E}[M]-\mathbb{E}[z]$, which is shown to be the maximal achievable rate for networks with Byzantine adversaries \cite{Jaggi2005a}. The computational cost of design, encoding, and decoding is dominated by the cost of carrying out the matrix multiplication $Y^{(i)}D^{(i)}$ in (\ref{xs}), which is $O(n(i\bar{c})^3)$. Notice that because $Y^{(i)}$ and $D^{(i)}$ grow regularly, $Y^{(i)}D^{(i)}$ is a block of $Y^{(i+1)}D^{(i+1)}$, therefore careful implementation of the algorithm can improve complexity (though not in the order sense) by building on the results from the last stage. Assume in general that $Y^{(i)}D^{(i)}=A^T$ and
\begin{equation*}
  \left(Y^{(i+1)}D^{(i+1)}\right)^T = \left[ \begin{array}{cc}
    A & C\\
    B & D
  \end{array} \right].
\end{equation*}
Then at stage $i+1$ we only need to perform the multiplications corresponding to blocks $B,\ C,$ and $D$.
 The same trick applies when we are to perform row reduction on $Y^{(i)}D^{(i)}$.
Suppose at the $i$-th stage we have already reduced $A$ into row echelon form with matrix $R$, \emph{i.e.}, with high probability it follows (otherwise there is a decoding error)
\begin{equation*}
  R A = A' = \left[\begin{array}{c}
    I\\\bm{0}
  \end{array} \right].
\end{equation*}
At the $(i+1)$-th stage we want to reduce $\left(Y^{(i+1)}D^{(i+1)}\right)^T$ into row echelon form based on the knowledge of $R$. We can construct the row operations as the following steps: i) multiply $R$ to reduce block $A$ and obtain $[A'\ RC]$ in the upper blocks; ii) use $A'$ to cancel block $B$ to zero; iii) perform row reduction on the lower right block corresponding to $D$; iv) use the row reduced lower right block to cancel the upper right block to zero. Formally, let
\begin{equation*}
  R'=\left[\begin{array}{cc}
    I & -RC|\bm{0}\\
    \bm{0} & I
  \end{array} \right]
  \left[\begin{array}{cc}
    I & \bm{0}\\
    \bm{0} & D^-
  \end{array} \right]
    \left[\begin{array}{cc}
    I & \bm{0}\\
    -B|\bm{0} & \bm{0}
  \end{array} \right]
      \left[\begin{array}{cc}
    R & \bm{0}\\
    \bm{0} & I
  \end{array} \right]
\end{equation*}
where $D^-$ is the matrix that row reduces $[B|\bm{0}]RC+D$, where the $\bm{0}$ after $B$ is defined as a zero padding sub-matrix of appropriate dimension. Then it follows
\begin{equation*}
  R' Y^{(i+1)}D^{(i+1)} = \left[ \begin{array}{cc}
    \begin{array}{c}
      I\\\bm{0}
    \end{array} & \bm{0}\\
    \hline
    \bm{0} & \begin{array}{c}
      I\\\bm{0}
    \end{array}
  \end{array} \right].
\end{equation*}
Finally we only need to permute the rows to place the identity matrix on top. Note that by this algorithm at every stage we only need to perform row reduction on a small block (corresponding to $D$) plus several multiplications.
\section{Code Construction for Random Secret Model}
In this section we consider the case that a secret channel is not available between Alice and Bob, instead they only share a ``small'' random secrets whose size is negligible compared to the amount of information sent. The random secret model is similar to the previously discussed secret channel model with the difference that the secret information should be random and independent with the source message $X_0$. Therefore in this case Alice cannot send the hash of $X_0$ to Bob secretly and reliably. Nutman and Langberg \cite{Nutman2008} modifies the scheme in \cite{Jaggi2008} against Byzantine adversaries under the secret channel model so that it also works  for the random secret model. The essential idea is to prefix $X_0$ with a small matrix $L$ obtained by solving the hash relation $(L\ X_0)D=H$, where both $D$ and $H$ are the fixed random secrets. Then $(L\ X_0)$ is sent through the network as the new input message. However, there is no obvious way to make this scheme rateless, because in the rateless setting, the number of columns of $D$ and $H$ will grow over time as more linearly independent redundancy is needed. But this implies it is not possible to uniquely solve $L$ given a fixed $X_0$.

The difficulty in solving $L$ reveals the restriction of defining the hash relation in the matrix form of $X_0D=H$. Hence a more flexible way for hashing the data is required as discussed in the following. Recall that the vectorization of a matrix is a linear transformation which converts the matrix into a column vector by stacking the columns of the matrix on top of one another. Let column vector $\bm{w} \in \mathbb{F}^{bn}_q$ be the vectorized\footnote{Compared to linear transforms on the matrix $W$, linear transforms on the vector $\bm{w}$ are more general (every operation in the former class can be representation by an operation in the latter class, but not vice versa)
} $W$ (recall that $W$ is the matrix of raw incompressible source data before attaching the identity matrix, defined in Section \ref{secenc}). %Notice that comparing to matrix $W$, vector $\bm{w}$ allows for more general and flexible linear operations.

To generate the hashes, \emph{i.e.}, the linearly independent redundancy that is transmitted at the $k$-th stage. We first draw $\alpha_k$ symbols from the random shared secrets as $d_1^{(k)}$, $d_2^{(k)}$, $...$, $d_{\alpha_k}^{(k)} \in \mathbb{F}_q$, and use them to construct the $\alpha_k \times nb$ parity check matrix $D_k=\left[d_{ij}^{(k)}\right]$, where
$d_{ij}^{(k)}=\left(d_i^{(k)}\right)^j$, $\forall 1 \le i \le \alpha_k$, $1 \le j \le nb$.
%$\alpha_k \times (nb+\sum_{i=1}^k\alpha_i)$ parity check matrix $D_k=[d_{ij}^{(k)}]$, where $\forall 1\le i \le \alpha_k$, \begin{equation*}
%  d_{ij}^{(k)}=\left\{ \begin{array}{ll} \left(d_i^{(k)}\right)^j &  1\le j\le nb + \sum_{l=1}^{k-1} \alpha_l + i\\
%  0 & nb + \sum_{l=1}^{k-1} \alpha_l + i < j \le nb + \sum_{l=1}^{k} \alpha_l
%  \end{array}\right.
%\end{equation*}

Then we draw another $\alpha_k$ symbols from the random shared secrets as $h_1^{(k)}, h_2^{(k)}, ..., h_{\alpha_k}^{(k)}$ and use them to construct the hash of the message. Let $\bm{h}_k = (h_1^{(k)}, ..., h_{\alpha_k}^{(k)})^T$. Then the following parity check relation is enforced:
\begin{align}\label{hashvec}
  [D_k\ \ I_{\alpha_k}]  \left[\begin{array}{l}
    \bm{w} \\ \bm{l}_k
  \end{array}\right] = \bm{h}_k,
\end{align}
where $I_{\alpha_k}$ is the identity matrix of dimension $\alpha_k$ and $\bm{l}_k$ is a column vector of length $\alpha_k$ that can be solved uniquely from (\ref{hashvec}) as
\begin{align*}
  \bm{l}_k = \bm{h}_k - D_k\bm{w}
\end{align*}

Note that $\{\bm{l}_k\}$ does not need to be kept secret. Now we can readily construct a rateless parity check scheme based on (\ref{hashvec}):
\begin{align}\label{hashrl}
  \left[ \begin{array}{lllll}
  D_1& I_{\alpha_1}&0&...&0\\
  {D_2}& 0 & I_{\alpha_2} & ...&0\\
  \vdots & & \vdots &\\
  D_i & 0 & 0 & ... & I_{\alpha_i}
  \end{array} \right]
  \left[ \begin{array}{c}
    \bm{w}\\ \bm{l}_1\\\vdots\\\bm{l}_i
  \end{array} \right] =
  \left[ \begin{array}{c}
    \bm{h}_1\\\vdots\\\bm{h}_i
  \end{array} \right],
\end{align}
\emph{i.e.}, the total number of parity checks $\sum_i \alpha_i$ can grow over time if necessary.

(\ref{hashrl}) implies that we are hashing $(\bm{w}^T\ \bm{l}^T)^T$ instead of the message $\bm{w}$ itself. However, the advantage of introducing $\bm{l}_k$ is that by attaching a short suffix to $\bm{w}$, we can establish a virtual secret channel between Alice and Bob. %: it appears that the random shared secrets $D_k$ and $\bm{h}_k$ were dependent on the message $\bm{w}$ now.
 Nevertheless, the challenge is then we need to send not only $\bm{w}$ but also $\{\bm{l}_k\}$ over the unreliable network publicly. In the following we will discuss the structure of the encoder and decoder for transmitting them successfully.

\subsection{Encoder}
In order for Bob to decode successfully, both linearly dependent redundancy and linearly independent redundancy are required. Redundant information that lies in the row space of $X_0$ is called linearly dependent redundancy and is used to combat erasures in the network (deletion from the row space of $X_0$). Other redundant information is linearly independent redundancy and is used to distill ``fake information'' that adversaries inject into the network (addition to the row space of $X_0$).
In the case of the secret channel model, we send linearly dependent redundancy $\{X_i\}$ in the network and send linearly independent redundancy $\{D_i\}$ and $\{H_i\}$ on the secret channel. Compared to that, in the case of the random secret model, a secret channel is not available and both linearly dependent and independent redundancy must be sent through the public and unreliable network.

However, notice that the linearly dependent redundancy corresponds to long messages that are usually arranged into long packets, while the linearly independent redundancy is short. Its size is chosen to be independent with $n$, as it is desired that the amount of random secrets required is negligible comparing to the amount of information sent. Therefore, it is convenient to encode and send the two kinds of redundancy separately (for example, it would be wasteful of resources if linearly independent redundancy is sent in normal packets because it is too short to fill a packet) as long packets and short packets, respectively. We define $M_i$, $z_i$, $c_i$, $\bar{c}$ for long packets as described in Section \ref{secmodel}.
For short packets, denote $\bar{M}_i$, $\bar{c}_i$ and $\bar{z}_i$ as the min cut from Alice to Bob, the number of available transmission opportunities, and the min cut from Calvin to Bob at stage $i$, respectively. Similarly we assume $\bar{z}_i<\bar{M}_i,\ \forall i$.

We first discuss the encoding scheme for linearly dependent redundancy. % which is similar to the scheme proposed for the secret channel model in Section \ref{secenc}
The source input message is arranged as a $b\times n$ matrix $W$. %We precode each row of $W_0$ with a $(n+b(\bar{c}-1),n)$ erasure correction MDS code, and arrange them into a $b\times (n+ b \bar{c}-b)$ matrix $W$.
Then we encode $X_0 =(W\  I_b)$. At the $i$-th stage, Alice draws a random matrix $K_i \in \mathbb{F}^{c_i \times b}_q$, and encodes the long packets $X_i = K_i X_0$.

To generate the linearly independent redundancy, at stage $i$ Alice sets $\alpha_i = i\sigma m$ (the choices of $\sigma$ and $m$ are discussed in the next paragraph), solves $\bm{l_i}$ according to (\ref{hashvec}), and arranges the column vector into a $\sigma \times i m$ matrix $\mathcal{L}_i$. Then she let $L_j = (\mathcal{L}_j\ \ \bm{0}_D\ \bm{0}_j\  I_{\sigma})$, $1 \le j \le i$, where $\bm{0}_D$ is a zero matrix of size $\sigma \times (i-j)m$, and $\bm{0}_j$ is the zero matrix of size $\sigma \times (j-1)\sigma$. $\bm{0}_D$ is dummy and is used to align $\mathcal{L}$, and $\bm{0}_j$ is used to align the identity matrix. %Notice that the size of $L_i$ is growing over time as we are introducing more linearly independent redundancy at every stage.
Alice then draws a random matrix $G_i$ of size $\bar{c}_i \times i \sigma$ and encodes the short packets as
\begin{align*}
  A_i = G_i \left[\begin{array}{llll}
  L_1&0&...&0\\\cline{1-2}
  \multicolumn{2}{c}{L_2}&...&0\\%\cline{1-3}
  \multicolumn{4}{c}{\cdots}\\\cline{1-4}
  \multicolumn{4}{c}{L_i}
  \end{array} \right] = G_i L^{(i)}.
\end{align*}

In order to eliminate the ``fake'' information injected by the adversaries, Alice should introduce an adequate amount of linearly independent redundancy, \emph{i.e.}, choosing $\{\alpha_i\}$ appropriately. Alice may choose any $\sigma$ such that $\sigma \le \bar{M}_i - \bar{z}_i,\ \forall i$ (e.g., $\sigma=1$ is a safe choice). She then chooses $m$ such that $\sigma m \ge 2b\bar{c} + 2\sigma\bar{c} +1$. Note that the size of the secret, $i(i+1)\sigma m /2$, is again negligible in $n$.

Finally, at the $i$-th stage Alice sends $X_i$ as long packets with packet length $n+b$, and $A_i$ as short packets with packet length $i(m+ \sigma)$. Alice repeats this procedure until Bob decodes successfully.

\subsection{Decoder}
At the $i$-th stage, Bob receives long packets $Y_i$ and short packets $J_i$:
\begin{align}
  Y_i & = T_i X_i + Q_i Z_i \label{Y_i2},\\
  J_i & = \bar{T}_i A_i + \bar{Q}_i E_i\label{J_i2},
\end{align}
where $T_i \in \mathbb{F}_{q}^{M_i \times c_i}$, $\bar{T}_i \in \mathbb{F}_{q}^{\bar{M}_i \times \bar{c}_i}$ are the transfer matrices between Alice and Bob, $Q_i \in \mathbb{F}_q^{M_i \times z_i},\ \bar{Q}_i \in  \mathbb{F}_q^{\bar{M}_i \times \bar{z}_i}$ are the transfer matrices between Calvin and Bob, and $Z_i \in \mathbb{F}_q^{z_i \times (n+b)}$, $E_i \in \mathbb{F}_q^{\bar{z}_i \times i(m+ \sigma)}$ are the errors injected to long packets and short packets, respectively.

Bob then stacks the long and short packets that he has received so far to get
\begin{equation*}
  Y^{(i)} = \left[ \begin{array}{c}
    Y_1\\
    \vdots\\
    Y_i
  \end{array} \right], \ \ \
  J^{(i)} = \left[ \begin{array}{llll}
  J_1&0&...&0\\\cline{1-2}
  \multicolumn{2}{c}{J_2}&...&0\\%\cline{1-3}
  \multicolumn{4}{c}{\cdots}\\\cline{1-4}
  \multicolumn{4}{c}{J_i}
  \end{array} \right],
\end{equation*}
where dummy matrix $\bm{0}_D$ is padded to $\{J_k\}$ in the same way as it is padded to $\{L_k\}$.

Bob evaluates the rank of $Y^{(i)}$, and waits to receive more packets if $r_i=$Rank$(Y^{(i)}) < b$. When $r_i \ge b$, Bob tries to decode. Without loss of generality we assume %that $r_i=\sum_{j=1}^{i} M_j$, \emph{i.e.},
the rows of $Y^{(i)}$ are linearly independent. Otherwise, % if $r_i<\sum_{j=1}^{i} M_j$,
Bob selects $r_i$ linearly independent rows from $Y^{(i)}$ and proceeds similarly. He then picks a basis for the column space of $Y^{(i)}$. As will be shown later, the last $b$ columns of $Y^{(i)}$ (corresponding to the identity matrix in $X_0$) are linearly independent $w.h.p.$, so they are chosen and is denoted by a $r_i \times b$ matrix $\hat{T}^{(i)}$. Without loss of generality (by permuting the columns if necessary) we assume that the remaining $r_i - b$ linearly independent columns correspond to the first $r_i - b$ columns of $Y^{(i)}$, denoted by a $r_i \times (r_i - b)$ matrix $T''^{(i)}$. So we can expand $Y^{(i)}$ with respect to this basis as
\begin{align}\label{Yi}
  Y^{(i)} = [ T''^{(i)}\ \hat{T}^{(i)} ] \left[ \begin{array}{lll}
    I_{r_i-b} & F^{Z} & 0\\
    0 & F^{X} & I_b
  \end{array} \right],
\end{align}
where $F^Z$ and $F^X$ are matrices of coefficients.

Bob deals with $J^{(i)}$ in a similar way. Let $\bar{r}_i$ be the rank of $J^{(i)}$, $\hat{\bar{T}}^{(i)} \in \mathbb{F}_q^{\bar{r}_i \times i\sigma}$ be the last $i\sigma$ columns of $J^{(i)}$, and $\bar{T}''^{(i)} \in \mathbb{F}_q^{ \bar{r}_i \times (\bar{r}_i-i\sigma)}$ be the first $\bar{r}_i-i\sigma$ columns of $J^{(i)}$. Then $w.h.p.$ $[\bar{T}''^{(i)}\  \hat{\bar{T}}^{(i)}]$ consists a basis of the column space of $J^{(i)}$, and we can write
\begin{align}\label{Ji}
  J^{(i)} = [\bar{T}''^{(i)}\  \hat{\bar{T}}^{(i)}] \left[ \begin{array}{lll}
    I_{\bar{r}_i-i \sigma} & F^{E} & 0\\
    0 & F^{A} & I_{i \sigma}
  \end{array} \right],
\end{align}
where $F^E$ and $F^A$ are matrices of coefficients.

Equations (\ref{Yi}) and (\ref{Ji}) characterize the relationship between the received observations and the input messages due to the effect of the network transform. In order to decode successfully, Bob needs to take into account the built-in redundancy of the message, \emph{i.e.}, the relation between $\bm{w}$ and $\{\bm{l}_i\}$, as follows. $\forall i$, split $X_0$ and $L^{(i)}$ as:
\begin{align}
  X_0& = [X_a^{(i)} \ X_b^{(i)} \ X_c^{(i)}]\label{cutx},\\
  L^{(i)}& = [L^{(i)}_a \ L^{(i)}_b \ L^{(i)}_c]\label{cutl},
\end{align}
where $X_a^{(i)}$ are the first $r_i - b$ columns of $X_0$, $X_c^{(i)}$ are the last $b$ columns of $X_0$, and $X_b^{(i)}$ are the remaining columns in the middle; $L^{(i)}_a$ are the first $\bar{r}_i - i\sigma$ columns of $L^{(i)}$, $L^{(i)}_c$ are the last $i \sigma$ columns of $L^{(i)}$, and $L^{(i)}_b$ are the remaining columns in the middle. Let $\bm{x}_a^{(i)}$, $\bm{x}_b^{(i)}$ and $\bm{x}_c^{(i)}$ be the vectorized versions of $X^{(i)}_a$, $X^{(i)}_b$ and $X^{(i)}_c$. Let $\bm{l}_a^{(i)}$, $\bm{l}_b^{(i)}$ and $\bm{l}_c^{(i)}$ be the vectorized versions of $L^{(i)}_a$, $L^{(i)}_b$ and $L^{(i)}_c$ omitting the dummy $\bm{0}_D$. By construction it follows that,
\begin{align}\label{xwcor}
  \left[ \begin{array}{c}
    \bm{x}^{(i)}_a\\\bm{x}^{(i)}_b\\\hline \bm{l}^{(i)}_a\\\bm{l}^{(i)}_b
  \end{array} \right]
  = \left[ \begin{array}{cc}
    \bm{w} \\\hline \bm{l}_1\\\vdots\\\bm{l}_i
  \end{array} \right].
\end{align}\
Note that $\bm{x}_c^{(i)}$ and $\bm{l}_c^{(i)}$ are left out because they correspond to the redundant identity matrix.
Now Bob constructs two matrix $B_{top}$ and $B_{mid}$ as defined in (\ref{Btop}) and (\ref{Bmid}), respectively. Here $f^Z_{i,j}$ and $f^E_{i,j}$ are the $(i,j)^{th}$ entries of matrix $F^Z$ and $F^E$, respectively, and $\beta=n+b-r_i$, $\gamma=i(m+\sigma)-\bar{r}_i$. He then deletes all columns in $B_{mid}$ corresponding to the positions of the dummy zero padding when vectoring $L^{(i)}$, and obtain a submatrix $B'_{mid}$.
\newcounter{tempequationcounter}
\begin{figure*}[!t]
\normalsize
\setcounter{tempequationcounter}{\value{equation}}
\addtocounter{tempequationcounter}{1}
\begin{IEEEeqnarray}{rCl}
\setcounter{equation}{\value{tempequationcounter}}
  B_{top} = \left[ \begin{array}{cccc|cccc}
    -f_{1,1}^Z \hat{T}^{(i)} &  -f_{2,1}^Z \hat{T}^{(i)} & ... & -f_{r_i-b,1}^Z\hat{T}^{(i)} & \hat{T}^{(i)} & 0 & ... & 0\\
    -f_{1,2}^Z \hat{T}^{(i)} &  -f_{2,2}^Z \hat{T}^{(i)} & ... & -f_{r_i-b,2}^Z\hat{T}^{(i)} & 0 & \hat{T}^{(i)} & ... & 0\\
    \vdots & \vdots & \vdots & \vdots & \vdots & \vdots & \vdots & \vdots\\
    -f_{1,\beta}^Z \hat{T}^{(i)} &  -f_{2,\beta}^Z \hat{T}^{(i)} & ... & -f_{r_i-b,\beta}^Z\hat{T}^{(i)} & 0 & 0 & ... & \hat{T}^{(i)}
  \end{array}
   \right]
\label{Btop}
\end{IEEEeqnarray}

\addtocounter{tempequationcounter}{1}
\begin{IEEEeqnarray}{rCl}
\setcounter{equation}{\value{tempequationcounter}}
  B_{mid} = \left[ \begin{array}{cccc|cccc}
    -f_{1,1}^E \hat{\bar{T}}^{(i)} &  -f_{2,1}^E \hat{\bar{T}}^{(i)} & ... & -f_{\bar{r}_i-i\sigma,1}^E\hat{\bar{T}}^{(i)} & \hat{\bar{T}}^{(i)} & 0 & ... & 0\\
    -f_{1,2}^E \hat{\bar{T}}^{(i)} &  -f_{2,2}^E \hat{\bar{T}}^{(i)} & ... & -f_{\bar{r}_i-i\sigma,2}^E\hat{\bar{T}}^{(i)} & 0 & \hat{\bar{T}}^{(i)} & ... & 0\\
    \vdots & \vdots & \vdots & \vdots & \vdots & \vdots & \vdots & \vdots\\
    -f_{1,\gamma}^E \hat{\bar{T}}^{(i)} &  -f_{2,\gamma}^E \hat{\bar{T}}^{(i)} & ... & -f_{\bar{r}_i-i\sigma,\gamma}^E\hat{\bar{T}}^{(i)} & 0 & 0 & ... & \hat{\bar{T}}^{(i)}
  \end{array}
   \right]
\label{Bmid}
\end{IEEEeqnarray}

\addtocounter{tempequationcounter}{1}
\addtocounter{tempequationcounter}{-1}
\setcounter{equation}{\value{tempequationcounter}}
\hrulefill
\vspace*{4pt}
\end{figure*}
Finally, Bob let
\begin{align*}
  B_{bot}=  \left[ \begin{array}{lllll}
  D_1& I_{\alpha_1}&0&...&0\\
  {D_2}& 0 & I_{\alpha_2} & ...&0\\
  \vdots & & \vdots &\\
  D_i & 0 & 0 & ... & I_{\alpha_i}
  \end{array} \right],
\end{align*}
Notice that if Bob permutes the columns of $Y^{(i)}$ and $J^{(i)}$ when constructing $T''^{(i)}$ and $\bar{T}''^{(i)}$, then he needs to permute the columns of $B_{bot}$ accordingly. Then he tries to solve the equations:
\begin{align}\label{keyeqn}
  B \left[ \begin{array}{c}
    \bm{x}^{(i)}_a\\\bm{x}^{(i)}_b\\\bm{l}^{(i)}_a\\\bm{l}^{(i)}_b
  \end{array} \right] = \left[ \begin{array}{c}
  \hat{\bm{T}}^{(i)} \bm{f}^X  \\ \hat{\bar{\bm{T}}}^{(i)} \bm{f}^A \\ \bm{h_1} \\ \vdots \\ \bm{h_i}
  \end{array} \right],
\end{align}
where $\bm{f}^X$, $\bm{f}^A$ are the vectorized versions of $F^X$, $F^A$, respectively, $\hat{\bm{T}}^{(i)} = \text{diag}[\hat{T}^{(i)},...,\hat{T}^{(i)}]$, $\hat{\bar{\bm{T}}}^{(i)} = \text{diag}[\hat{\bar{T}}^{(i)},...,\hat{\bar{T}}^{(i)}]$, and the matrix $B$ is defined as:
\begin{align*}
  B = \left[ \begin{array}{cc}
    B_{top} & 0\\
    0 & B'_{mid}\\
    \hline
    \multicolumn{2}{c}{B_{bot}}
  \end{array} \right].
\end{align*}

Bob tries to solve (\ref{keyeqn}) and, if there exists no solution, he waits for more redundancy from Alice and tries to solve it again at the next stage. If there is a unique solution to (\ref{keyeqn}), then Bob has decoded successfully with high probability. Otherwise, if there are multiple solutions, Bob declares a decoding failure.

\subsection{Performance}
In this section we show the proposed scheme will succeed with high probability and achieve the optimal rate. Our first step is to establish (\ref{Yi}) and (\ref{Ji}).
We first consider the short packets. Note that (\ref{Ji}) is shown by Lemma \label{tfullcol} below.
\begin{lemma}\label{tfullcol}
$\hat{\bar{T}}^{(i)}$ has full column rank with high probability.
\end{lemma}
\begin{proof} For notational convenience we define
\begin{align*}
  \bar{T}^{(i)} = \left[ \begin{array}{llll}
    \bar{T}_1 & 0 & ... & 0\\
    0 & \bar{T_2} & ... & 0\\
    \multicolumn{4}{c}{...}\\
    0 & 0 & ... & \bar{T}_i
  \end{array}
  \right],\
      A^{(i)} = \left[ \begin{array}{llll}
  A_1&0&...&0\\\cline{1-2}
  \multicolumn{2}{c}{A_2}&...&0\\%\cline{1-3}
  \multicolumn{4}{c}{\cdots}\\\cline{1-4}
  \multicolumn{4}{c}{A_i}
  \end{array} \right]
\end{align*}
\begin{align*}
  \bar{Q}^{(i)} = \left[ \begin{array}{llll}
    \bar{Q}_1 & 0 & ... & 0\\
    0 & \bar{Q}_2 & ... & 0\\
    \multicolumn{4}{c}{...}\\
    0 & 0 & ... & \bar{Q}_i
  \end{array}
  \right],\
      E^{(i)} = \left[ \begin{array}{llll}
  E_1&0&...&0\\\cline{1-2}
  \multicolumn{2}{c}{E_2}&...&0\\%\cline{1-3}
  \multicolumn{4}{c}{\cdots}\\\cline{1-4}
  \multicolumn{4}{c}{E_i}
  \end{array} \right]
\end{align*}
%\begin{align*}
%    A^{(i)} = \left[ \begin{array}{lllll}
%  A_1&0&0&...&0\\\cline{1-2}
%  \multicolumn{2}{c}{A_2}&0&...&0\\\cline{1-3}
%  \multicolumn{3}{c}{A_3}&...&0\\%\cline{1-4}
%  \multicolumn{5}{c}{\cdots}\\\cline{1-5}
%  \multicolumn{5}{c}{A_i}
%  \end{array} \right]
%\end{align*}
%\begin{align*}
%    E^{(i)} = \left[ \begin{array}{lllll}
%  E_1&0&0&...&0\\\cline{1-2}
%  \multicolumn{2}{c}{E_2}&0&...&0\\\cline{1-3}
%  \multicolumn{3}{c}{E_3}&...&0\\%\cline{1-4}
%  \multicolumn{5}{c}{\cdots}\\\cline{1-5}
%  \multicolumn{5}{c}{E_i}
%  \end{array} \right]
%\end{align*}
Then we have the concise relationship from (\ref{Y_i2}) and (\ref{J_i2}):
\begin{align}
  J^{(i)} &= \bar{T}^{(i)} A^{(i)} + \bar{Q}^{(i)} E^{(i)}\nonumber\\
          &= \bar{T}^{(i)} G^{(i)} L^{(i)} + \bar{Q}^{(i)} E^{(i)} \label{Jifull},
\end{align}
where
\begin{equation*}
      G^{(i)} = \left[ \begin{array}{llll}
  G_1&0&...&0\\\cline{1-2}
  \multicolumn{2}{c}{G_2}&...&0\\%\cline{1-3}
  \multicolumn{4}{c}{\cdots}\\\cline{1-4}
  \multicolumn{4}{c}{G_i}
  \end{array} \right].
\end{equation*}
By construction, $\sum_{j=1}^i \bar{M}_j - \sum_{j=1}^i \bar{z}_j \ge i \sigma$. By an argument identical to Lemma \ref{disjoint}, this implies that with high probability $\bar{T}^{(i)}$ and $\bar{Q}^{(i)}$ both have full column rank and span disjoint column spaces (except for the zero vector). Random matrix $G^{(i)}$ also has full rank \emph{w.h.p.} and therefore $\bar{T}^{(i)} G^{(i)}$ and $\bar{Q}^{(i)}$ also both have full column rank and span disjoint column spaces. We can write the last $i \sigma$ columns in (\ref{Jifull}) corresponding to the redundant identity matrix in $L^{(i)}$ as
\begin{align}\label{tcap}
  \hat{\bar{T}}^{(i)} = \bar{T}^{(i)} G^{(i)} + \bar{Q}^{(i)} E^{(i)}_{\text{rear}},
\end{align}
where $E^{(i)}_{\text{rear}}$ are the last $i\sigma$ columns of $E^{(i)}$. Hence the columns of $\hat{\bar{T}}^{(i)}$ are linearly independent \emph{w.h.p.}.
\end{proof}
A similar argument also holds for the long packets:
\begin{lemma}
  If $\sum_{j=1}^i M_j - \sum_{j=1}^i z_j \ge b$, then $\hat{{T}}^{(i)}$ has full column rank with high probability.
\end{lemma}
\begin{proof}
  Consider the last $b$ columns in (\ref{first}) corresponding to the redundant identity matrix in $X_0$
  \begin{align*}
    \hat{T}^{(i)} = T^{(i)} + Q^{(i)} Z^{(i)}_{\text{rear}}.
  \end{align*}
  And by Lemma (\ref{disjoint}), $T^{(i)}$ and $Q^{(i)}$ both have full column rank and span disjoint column spaces (except for the zero vector). Hence the columns of $\hat{{T}}^{(i)}$ are linearly independent.
\end{proof}

Now we are ready to analyze the key equation (\ref{keyeqn}). We first prove a related lemma.
\begin{lemma}
  With high probability (\ref{J_i2}) and (\ref{Ji}) are equivalent to the following equation:
  \begin{equation}\label{tlmatch}
    \hat{\bar{T}}^{(i)} L_b^{(i)} = \hat{\bar{T}}^{(i)} ( F^{A} + L^{(i)}_a F^E).
  \end{equation}
\end{lemma}
\begin{proof}
  We use a technique similar to \cite{Jaggi2008}. Substituting (\ref{tcap}) to (\ref{Jifull}), it follows
  \begin{equation*}
    J^{(i)} = \hat{\bar{T}}^{(i)} L^{(i)} + \bar{Q}^{(i)} (E^{(i)} - E^{(i)}_{\text{rear}} L^{(i)} ).
  \end{equation*}
  Then by (\ref{Ji}) we have:
  \begin{multline}\label{compare}
    \hat{\bar{T}}^{(i)} L^{(i)} + \bar{Q}^{(i)} (E^{(i)} - E^{(i)}_{\text{rear}} L^{(i)} )
    =\\ \hat{\bar{T}}^{(i)} [0\ F^A\ I_{i \sigma} ] + \bar{T}''^{(i)} [I_{\bar{r}_i - i \sigma}\ F^E\ 0].
  \end{multline}

  Therefore the columns of $\bar{T}''^{(i)}$ are spanned by the columns of $[\hat{\bar{T}}^{(i)}\ \bar{Q}^{(i)}]$. So there exists matrices $V_1$ and $V_2$ such that
  \begin{equation*}
    \bar{T}''^{(i)} = \hat{\bar{T}}^{(i)} V_1 + \bar{Q}^{(i)} V_2.
  \end{equation*}
  And we can rewrite (\ref{compare}) as
    \begin{multline}\label{compare2}
    \hat{\bar{T}}^{(i)} L^{(i)} + \bar{Q}^{(i)} (E^{(i)} - E^{(i)}_{\text{rear}} L^{(i)} )
    =\\ \hat{\bar{T}}^{(i)} [0\ F^A\ I_{i \sigma} ] + (\hat{\bar{T}}^{(i)} V_1 + \bar{Q}^{(i)} V_2) [I_{\bar{r}_i - i \sigma}\ F^E\ 0].
  \end{multline}
  However, notice from (\ref{tcap}) that the columns of $\hat{\bar{T}}^{(i)}$ and the columns of $\bar{Q}^{(i)}$ are $w.h.p.$ linearly independent, \emph{i.e.}, the column spaces of $\hat{\bar{T}}^{(i)}$ and $\bar{Q}^{(i)}$ are disjoint except for the zero vector. Hence (\ref{compare2}) implies the following set of equations:
  \begin{align}\label{compared}
    \hat{\bar{T}}^{(i)} L^{(i)} = \hat{\bar{T}}^{(i)}[0\ F^A\ I_{i \sigma} ] + \hat{\bar{T}}^{(i)} V_1 [I_{\bar{r}_i - i \sigma}\ F^E\ 0]\\
    \bar{Q}^{(i)} (E^{(i)} - E^{(i)}_{\text{rear}} L^{(i)} ) = \bar{Q}^{(i)} V_2 [I_{\bar{r}_i - i \sigma}\ F^E\ 0].
  \end{align}
  Here (\ref{compared}) suffices for the purpose of decoding $X_0$. We split (\ref{compared}) into three parts as in (\ref{cutl}), and get:
  \begin{align}
    \hat{\bar{T}}^{(i)} L^{(i)}_a &= \hat{\bar{T}}^{(i)} V_1 \label{ca}\\
    \hat{\bar{T}}^{(i)} L^{(i)}_b &= \hat{\bar{T}}^{(i)} F^A + \hat{\bar{T}}^{(i)} V_1 F^E\label{cb}\\
    \hat{\bar{T}}^{(i)} L^{(i)}_c &= \hat{\bar{T}}^{(i)}\label{cc}
  \end{align}
  By Lemma \ref{tfullcol}, $\hat{\bar{T}}^{(i)}$ has full column rank and is left-invertible, therefore by (\ref{ca}) it follows $L^{(i)}_a = V_1$. Substituting it into (\ref{cb})
  \begin{align*}
    \hat{\bar{T}}^{(i)} L^{(i)}_b &= \hat{\bar{T}}^{(i)} F^A + \hat{\bar{T}}^{(i)} L^{(i)}_a F^E.
  \end{align*}
  Finally, notice that by construction $L^{(i)}_c$ is the identity matrix and therefore (\ref{cc}) is redundant. Hence we can conclude that (\ref{J_i2}) and (\ref{Ji}) are equivalent to $$\hat{\bar{T}}^{(i)} L^{(i)}_b = \hat{\bar{T}}^{(i)} (F^A + L^{(i)}_a F^E)$$.
\end{proof}

By a similar argument, for long packets we have the following result:
\begin{lemma}
  If $\sum_{j=1}^i M_j - \sum_{j=1}^i z_j \ge b$, then with high probability (\ref{Y_i2}) and (\ref{Yi}) are equivalent to
    \begin{equation}\label{txmatch}
    \hat{{T}}^{(i)} X_b^{(i)} = \hat{{T}}^{(i)} ( F^{X} + X^{(i)}_a F^Z).
  \end{equation}
\end{lemma}
\begin{cor}\label{keyeqnproof}
  If $\sum_{j=1}^i M_j - \sum_{j=1}^i z_j \ge b$, then the matrix equation (\ref{keyeqn}) holds with high probability.
\end{cor}
\begin{proof}
  Notice that (\ref{keyeqn}) are equivalent to the following set of three matrix equations:
  \begin{align}\label{txmatch2}
    B_{top} \left[\begin{array}{c}
              \bm{x}^{(i)}_a\\
              \bm{x}^{(i)}_b
    \end{array}\right] & = \hat{\bm{T}}^{(i)} \bm{f}^X\\\label{tlmatch2}
    B'_{mid} \left[\begin{array}{c}
              \bm{l}^{(i)}_a\\
              \bm{l}^{(i)}_b
    \end{array}\right] & = \hat{\bar{\bm{T}}}^{(i)} \bm{f}^A\\\label{hashmatch2}
      B_{bot} \left[ \begin{array}{c}
    \bm{x}^{(i)}_a\\\bm{x}^{(i)}_b\\\bm{l}^{(i)}_a\\\bm{l}^{(i)}_b
  \end{array} \right] & = \left[ \begin{array}{c}
   \bm{h_1} \\ \vdots \\ \bm{h_i}
  \end{array} \right]
  \end{align}
  But now note that (\ref{txmatch2}) is equivalent to (\ref{txmatch}); (\ref{hashmatch2}) is equivalent to (\ref{hashrl}); and (\ref{tlmatch2}) is equivalent to (\ref{tlmatch}) because all the deleted columns correspond to the zero padding in $L^{(i)}$.
\end{proof}

%Finally we need to prove that under appropriate conditions equation (\ref{keyeqn}) has a unique solution. Lemma \ref{noX'} below shows the probability of decoding an error packet is vanishing:
Finally we need to prove that (\ref{keyeqn}) has a unique solution, \emph{i.e.}, the probability of decoding an error packet is vanishing:
\begin{lemma}\label{noX'}
  If $\sigma m \ge 2b\bar{c} + 2\sigma\bar{c}+1$, then with high probability there does not exist $X' \ne X_0$ such that $X'$ satisfies (\ref{keyeqn}).
\end{lemma}
\begin{proof}
    %By  Corollary \ref{keyeqnproof}, if $\sum_{j=1}^i M_j - \sum_{j=1}^i z_j \ge b$, then $w.h.p.$ $X_0$ can be solved from (\ref{keyeqn}), and by Lemma \ref{fullrank}, this solution is unique.
   %So we only need to consider the case that $\sum_{j=1}^i M_j - \sum_{j=1}^i z_j < b$.
   Suppose $X' \ne X_0$, and let $\bm{x}'_a$, $\bm{x}'_b$ be its vectorized versions as described in (\ref{cutx}). We consider the probability that there exist $\bm{x}'_a$, $\bm{x}'_b$, $\bm{l}'_a$ and $\bm{l}'_b$ that satisfy (\ref{keyeqn}). Let us first consider the top $\beta r_i + \gamma \bar{r}_i$ rows in $B$ that correspond to the blocks of $B_{top}$ and $B'_{mid}$
      \begin{align}\label{topkey}
      \left[ \begin{array}{cc}
        B_{top} & \bm{0}\\
        \bm{0} & B'_{mid}
       \end{array} \right] \left[ \begin{array}{c}
    \bm{x}'^{(i)}_a\\\bm{x}'^{(i)}_b\\\bm{l}'^{(i)}_a\\\bm{l}'^{(i)}_b
  \end{array} \right] = \left[ \begin{array}{c}
  \hat{\bm{T}}^{(i)} \bm{f}^X  \\ \hat{\bar{\bm{T}}}^{(i)} \bm{f}^A
  \end{array} \right],
   \end{align}
   They are equivalent to
   \begin{align}
     \label{x'}
     X_b'^{(i)} & =  F^{X} + X'^{(i)}_a F^Z\\\label{l'}
     L_b'^{(i)} & =  F^{A} + L'^{(i)}_a F^E
   \end{align}
   Therefore given arbitrary values of $\bm{x}_a'^{(i)}$ and $\bm{l}_a'^{(i)}$, there are unique corresponding values of $\bm{x}_b'^{(i)}$ and $\bm{l}_b'^{(i)}$ that satisfy (\ref{topkey}).

     Now given any $\bm{x}_a'^{(i)}$ and $\bm{l}_a'^{(i)}$ (and the corresponding $\bm{x}_b'^{(i)}$ and $\bm{l}_b'^{(i)}$) such that (\ref{topkey}) holds, we consider the probability that the bottom $\sum_{k=1}^i\alpha_k = (i^2+i)\sigma m/2$ rows in (\ref{keyeqn}) also holds:
  \begin{align}\label{x'hash}
     B_{bot} \left[ \begin{array}{c}
    \bm{x}'^{(i)}_a\\\bm{x}'^{(i)}_b\\\bm{l}'^{(i)}_a\\\bm{l}'^{(i)}_b
  \end{array} \right] & = \left[ \begin{array}{c}
   \bm{h_1} \\ \vdots \\ \bm{h_i}
  \end{array} \right]
  \end{align}
  This is equivalent to:
     \begin{align}\label{x'hashe}
       B_{bot} \left[ \begin{array}{c}
    \bm{x}^{(i)}_a - \bm{x}'^{(i)}_a\\\bm{x}^{(i)}_b - \bm{x}'^{(i)}_b\\\bm{l}^{(i)}_a - \bm{l}'^{(i)}_a\\\bm{l}^{(i)}_b - \bm{l}'^{(i)}_b
  \end{array} \right] = \bm{0},
   \end{align}
     Because $X' \ne X_0$, so $\bm{x}^{(i)}_a-\bm{x}'^{(i)}_a$ and $\bm{x}^{(i)}_b-\bm{x}'^{(i)}_b$ cannot both be the zero vector. Denote
      \begin{align*}
        \bm{x}^{(i)}_a-\bm{x}'^{(i)}_a & = (x^{(i)}_{a,1},...,x^{(i)}_{a,\theta_a})^T\\
        \bm{x}^{(i)}_b-\bm{x}'^{(i)}_b & = (x^{(i)}_{b,1},...,x^{(i)}_{b,\theta_b})^T\\
        \left[\begin{array}{c}
          \bm{l}^{(i)}_a-\bm{l}'^{(i)}_a\\
          \bm{l}^{(i)}_b-\bm{l}'^{(i)}_b
        \end{array}
         \right] & = (l^{(i)}_{1},...,l^{(i)}_{\theta_l})^T
      \end{align*}
      where $\theta_a = b(r_i -b)$, $\theta_b = \beta b$ and $\theta_l =(i^2+i)\sigma m /2$. Denote the $(u,v)$ entry of $B_{bot}$ as $s_{u,v}$, then the $j$-th row of (\ref{x'hashe}) is
       \begin{multline}\label{poly}
         \sum_{k=1}^{b(r_i -b)} x^{(i)}_{a,k} s_{j,k} + \sum_{k=1}^{\beta b} x^{(i)}_{b,k} s_{j,k+b(r_i -b)}\\ + \sum_{k=1}^{(i^2+i)\sigma m /2} l^{(i)}_k s_{j,k+nb} = 0
       \end{multline}
       Let $s_j$ be the $(j,1)$ entry of $B_{bot}$ before column permutation, then $s_{j,k}=s_j^{\pi(k)}$, $1\le k\le nb$, where $\pi$ is a permutation of $\{1,...,nb\}$. So (\ref{poly}) is a non-zero polynomial of order at most $b(r_i -b) + \beta b = nb$ in variable $s_j$ (the $\{s_{j,k+nb}\}$ are constants 0 or 1 by construction and are independent with respect to $s_j$). By the fundamental theorem of algebra the polynomial have at most $nb$ roots. And the probability that $s_j$ is chosen as one of the roots is at most $nb/q$, and this is the upper bound of the probability that row $j$ holds in (\ref{x'hashe}). Because $\{s_{j}\}$ are chosen independently, (\ref{x'hashe}) holds with probability no larger than $(nb/q)^{(i^2+i)\sigma m/2}$.
       %Hence for a given pattern of $X'$, the probability that (\ref{x'hash}) holds is at most $(P_2/q)^{(i^2+i)\sigma m/2}$.

    Finally, there are at most $q^{b(r_i-b)}$ different $\bm{x}_a^{(i)}$ and at most $q^{i\sigma(\bar{r}_i-i\sigma)}$ different $\bm{l}_a^{(i)}$. By (\ref{Y_i2}), $r_i - b \le i\bar{c}$, and by (\ref{J_i2}), $\bar{r}_i - i\sigma \le i \bar{c}$. Therefore by the union bound, the probability that there exists $X'_0 \ne X_0$ such that $\bm{x}'_a$, $\bm{x}'_b$, $\bm{l}'_a$ and $\bm{l}'_b$ satisfy (\ref{keyeqn}) is at most
    \begin{align*}
      \left(\frac{nb}{q}\right)^{\frac{(i^2+i)\sigma m}{2}} q^{ib\bar{c}+i^2\sigma\bar{c}} \le \frac{(nb)^{i^2\sigma m}}{q^{i^2}} \to 0
    \end{align*}
  %Further, according to (\ref{x'}) there is an one-to-one correspondence between $X'$ and $X'^{(i)}_a$. So there are at most $q^{b(r_i-b)}$ different $X'$. Finally, by the union bound the probability that there exists a $X'$ such that (\ref{keyeqn}) holds is upper bounded by    $P_2^{(i^2+i)\sigma m/2}q^{i b \bar{c} - (i^2+i)\sigma m/2} < P_2^{(i^2+i)\sigma m/2}/q^{i^2} \to 0$.
\end{proof}

We are ready to present the final conclusion.
\begin{thm}\label{thmrs}
  $\forall i$ such that $b+\sum_{j=1}^{i} z_j \le \sum_{j=1}^{i} M_j$, with the proposed coding scheme, Bob is able to decode $X_0$ correctly with high probability at the $i$-th stage. Otherwise, Bob waits for more redundancy instead of decoding erroneous packets.
\end{thm}
\begin{proof}
By Corollary \ref{keyeqnproof}, $X_0$ can be solved from (\ref{keyeqn}) if $b+\sum_{j=1}^{i} z_j \le \sum_{j=1}^{i} M_j$. By Lemma \ref{noX'}, if a solution exists, it is correct and unique. Otherwise, there is no solution to (\ref{keyeqn}) and by the algorithm Bob waits for more redundancy. %Finally, because $z_i < M_i$, the algorithm will finish in at most $b$ stages, and the length of $\bm{x}_a^{(i)}$ is $b(r_i - b) \le b^2 (\bar{c} -1)$. Therefore $W_0$ can be reconstructed from $\bm{x}_b^{(i)}$ by the MDS code.
\end{proof}

Similar to the case of the secret channel model, Theorem \ref{thmrs} shows that our code is optimal in that sense that decoding succeeds with high probability whenever the total amount of
information received by the sink satisfies the cut set bound with
respect to the amount of message and error information.  We can also show rate-optimality under the  i.i.d. case.
\begin{thm}
  Assume $M_i$, $z_i$, $i=1,2...$ are i.i.d. random variables with mean $\mathbb{E}[M]$ and $\mathbb{E}[z]$, respectively.
 If there exists $\epsilon >0$ such that $\mathbb{E}[M]-\mathbb{E}[z] \ge \epsilon$, then with the proposed coding scheme Bob is able to decode $X_0$ correctly with high probability. And on average the code achieves rate $$r \ge \frac{b}{b+\bar{c}-1} \left(\mathbb{E}[M] - \mathbb{E}[z]\right).$$
\end{thm}
\begin{proof}
  Note that both long packets and short packets are sent over the network. We consider the short packets to be overhead. At the $i$-th stage, the length of a short packet is $i(m+\sigma)$, and is negligible as a large enough $n$ is chosen. The rest of the proof is identical to the proof of Theorem \ref{rate}.
\end{proof}

Again the proposed scheme is asymptotically rate optimal as we choose a large enough $b$. The computational cost of design, encoding, and decoding is dominated by the cost of solving (\ref{keyeqn}), which equals $O((ni\bar{c})^3)$.

\section{Conclusion}
This paper introduces information-theoretical rateless resilient network codes against Byzantine adversaries. Unlike previous works, knowledge about the network and adversaries are not required and the codes will adapt to their parameters by sending more redundancy over time if necessary. We present two algorithms
targeting two network models. The first model
assumes there is a low-rate secret channel between the source and the
destination. The
second model assumes the source and destination share some ``small'' random secrets that are
independent with the input information. For both models our codes are rate-optimal, distributed, polynomial-time, work on general topology, and only require source and destination nodes to be modified.





%\bibliographystyle{IEEEtran}

%\begin{thebibliography}{1}
%\end{thebibliography}

\bibliographystyle{IEEEtran}	% (uses file "plain.bst")
\bibliography{bare_conf}

\end{document}